\documentclass[12pt]{article}

\usepackage{hyperref}

\usepackage{amsmath} 
\usepackage{amssymb}
\usepackage{amsfonts}
\usepackage{amsthm}
\usepackage{graphicx}
\usepackage{xcolor}

\newtheorem{proposition}{Proposition}

\newtheorem{theorem}{Theorem}
\newtheorem{definition}{Definition}
\newtheorem{corollary}{Corollary}

\theoremstyle{definition}

\newtheorem{remark}{Remark}

\newcommand\Tr{\operatorname{Tr}}

\begin{document}

	\begin{center}
		\Large
		\textbf{On time-dependent projectors and on generalization of thermodynamical approach to open quantum systems}
		
		\large 
		\textbf{K.Sh. Meretukov}\footnote{Faculty of Physics, Lomonosov Moscow State University, Leninskie Gory, Moscow 119991, Russia\\
			E-mail:\href{mailto:violblink@gmail.com}{violblink@gmail.com}}
		\textbf{A.E. Teretenkov}\footnote{Department of Mathematical Methods for Quantum Technologies, Steklov Mathematical Institute of Russian Academy of Sciences, ul. Gubkina 8, Moscow 119991, Russia\\ E-mail:\href{mailto:taemsu@mail.ru}{taemsu@mail.ru}}
		\\[1mm]
	\end{center}
	
	\footnotesize
	In this paper, we develop a consistent perturbative technique for obtaining a time-local master equation based on projective methods in the case where the projector depends on time. We then introduce a generalization of the Kawasaki--Gunton projector, which allows us to use this technique to derive, generally speaking, nonlinear master equations in the case of arbitrary ansatzes consistent with some set of observables. Most of our results are very general, but in our discussion we focus on the application of these results to the theory of open quantum systems.
	\normalsize

	\section{Introduction}
	
	One of the most widely used approaches for deriving master equations in the theory of open quantum systems is the 	Nakajima--Zwanzig \cite{Nakajima1958, Zwanzig1960} projection method. In the case of the ''usual'' statement of the problem, when we are interested in the dynamics of the density matrix of an open system and it is considered to be ''slow'', the Argyres--Kelley projector is used \cite{Argyres1964}. Although in the theory of open quantum systems other projectors, which include   a part of the degrees of freedom of the environment \cite{Breuer2006, Breuer2007} or, conversely, only a part of the system density matrix \cite{Trushechkin2019}  in the ''slow'' dynamics are also used.

	However, all of the mentioned projectors are time-independent. At the same time, within the framework of nonequilibrium statistical physics a number of methods based on projectors which, in general, clearly depend on time \cite[Section 2C]{Zubarev1997}, have been developed, \cite{Rau1996}, in particular the Mori \cite{Mori1965}, Kawasaki--Gunton \cite{Kawasaki1973}, Robertson \cite{Robertson1966} projectors, and their variations and generalizations \cite{Bouchard2007}. Recently, the use of these methods, primarily based on the  Kawasaki--Gunton operator and its generalizations, in the theory of open quantum systems has also actively attracted the researchers' attention \cite{Semin2015, Semin2020}, although some discussion of application of these methods to open systems has arisen before \cite{Seke1980}. 
	
	In works on nonequilibrium statistical physics perturbative expansions with such projectors are  usually (see \cite[Section 17.3]{Fick1990}, \cite[Section 2.4.2]{Zubarev1997}) based on a generalization of the Nakajima--Zwanzig integro-differential equations to the case of a time-dependent projector. Although the ''classical'' derivation of these arrangements in the physical literature \cite[Section 3.3.1]{BreuerPetruccione2007} often implies that integro-differential equations are more accurate than time-local equations, some explicit examples \cite{Teretenkov2019Non} say this is not necessarily true. In addition, time-local equations are considered easier to solve and are more common in open quantum systems theory \cite{Shibata1977, Breuer1999, Breuer2001}. For these reasons, in Section \ref{sec:local} of this paper we focus specifically on the derivation of time-local master equations and their perturbative approximations. We also discuss the use of these results to obtain nonlinear time-local master equations.
	
	In Section \ref{sec:genKavGun} we focus on a special class of time-dependent projectors, namely, we develop a generalization of the  Kawasaki--Gunton projector. We propose to consider arbitrary ansatzes parameterized by averages from relevant observables. We study the properties of such a generalized Kawasaki--Gunton projector. Based on the results of Section \ref{sec:local}, we obtain explicit second-order equations with such a projector in the case of an initial condition consistent with the projector.
	
	In Section \ref{sec:examples} we give the simplest examples of the Kawasaki--Gunton projectors and generalized Kawasaki--Gunton projectors. In particular, we give an example where the Kawasaki--Gunton projector looks differently in the traditional approach in different parametrizations, but is reduced to the same, and time-independent, projector. Since the approach to open quantum systems in \cite{Semin2020} has been called thermodynamical, it is natural to call our approach as a generalization of the thermodynamical approach. We discuss in general how constant projectors arise as special cases of the Kawasaki--Gunton projector.  Furthermore, we give the simplest example of a Kawasaki-Gunton projector that is not reducible to a time-independent projector.
	
	In Section \ref{sec:examplesOfEq} we give examples of the second order equation obtained in~\ref{sec:genKavGun} in the case of projectors from Section \ref{sec:examples}. In particular, in the example of a time-independent projector we make an observation about the properties of these equations after the Bogolubov-van Hove scaling \cite{Bogoliubov1946, VanHove1954}, which can be interpreted as an opportunity to make an analog of the Wick rotation in open quantum systems theory. 
	
	In the Conclusion we summarize and propose a number of directions for further development of these results that we consider interesting.
	
	\section{Time-local equations with time-dependent projector}
	\label{sec:local}
	
	In this section we obtain time-local master equations for the projected dynamics in the case of a time-dependent projector, as well as their perturbative approximations in the form following \cite{Nestmann2019, Teretenkov2022, Karasev2023}. Namely, we generalize Theorem 1 of \cite{Karasev2023} to the case of a time-dependent projector.  In this paper, everywhere we consider only finite-dimensional matrices and mappings. For linear mappings between matrices we will use the term ''superoperator''. We will follow the physical tradition and use the word projector as a synonym of the word idempotent without assuming that the projector is selfadjoint with respect to some scalar product or other. In the case of an arbitrary superoperator $\mathcal{A}$ we will use the notation $\mathcal{A}^{(-1)}$ for a pseudo-inverse superoperator such that $\mathcal{A}^{(-1)} \mathcal{A} = \mathcal{P}$, $\mathcal{A}^{(-1)} \mathcal{Q} = \mathcal{Q} \mathcal{A}^{(-1)} = 0$, where $\mathcal{P} $ is the projector on the image of $\mathcal{A} $, and $\mathcal{Q} $ is the one on its kernel.

		We will assume that the function $\rho: [t_0, +\infty) \rightarrow \mathbb{C}^{d \times d} $, $d \in \mathbb{N}$, $t_0 \in \mathbb{R} $ satisfies the ordinary linear equation
		\begin{equation}\label{eq:basicDiffEq}
			\frac{d}{dt} \rho(t) = \lambda \mathcal{L}(t)\rho(t),
		\end{equation}
		where $\mathcal{L}(t)$ is a superoperator-valued function $ t\in [t_0, +\infty)$, meaning that for every fixed $ t$ the mapping $\mathcal{L}(t)$ is a linear mapping from $ \mathbb{C}^{d \times d}$ to $ \mathbb{C}^{d \times d}$, and $\lambda$ is a real parameter. We will also assume that $\mathcal{L}(t)$ is continuous at $ t\in [t_0, +\infty)$.
		
		To identify the degrees of freedom of the system for which we will write out the dynamics, we introduce the function $\mathcal{P}(t)$, which for every fixed $ t$ is a linear mapping from $ \mathbb{C}^{d \times d}$ to $ \mathbb{C}^{d \times d}$. We will assume that it is continuously differentiable at $ t\in [t_0, +\infty)$. Furthermore, we will assume that at each fixed $t$ the value of this function is idempotent, i.e., $ \mathcal{P}(t)= \mathcal{P}^2(t)$.  Let us denote $ \mathcal{Q}(t) \equiv I - \mathcal{P}(t)$, where $I$ is a unit superoperator. 
		
		\begin{theorem}\label{th:masterEquation}
			Let $\mathcal{U}_{t_0}^t$ be a superoperator-valued function that is a solution of the equation
			\begin{equation}\label{eq:propEq}
				\frac{d}{dt}\mathcal{U}_{t_0}^t = \lambda \mathcal{L}(t) \mathcal{U}_{t_0}^t, \qquad \mathcal{U}_{t_0}^{t_0} = I, 
			\end{equation}
			where $I$ is an identity mapping from $ \mathbb{C}^{d \times d}$ to $ \mathbb{C}^{d \times d}$, $\lambda$ is a real parameter. Then, when $\lambda$ is small enough (and $t$ is fixed), $\mathcal{P}(t) \rho(t)$ satisfies the linear kinetic equation
			\begin{equation}\label{eq:masterEquation}
				\frac{d}{dt} (\mathcal{P}(t) \rho(t)) = \mathcal{K}(t)\mathcal{P}(t) \rho(t) + \mathcal{I}(t) \mathcal{Q}(t) \rho(t_0) , 
			\end{equation}
			where
			
			\begin{align}
				\mathcal{K}(t) &\equiv \left( \dot{\mathcal{P}}(t) \mathcal{U}_{t_0}^t\mathcal{P}(t) +   \mathcal{P}(t)\dot{\mathcal{U}}_{t_0}^t \mathcal{P}(t)\right) (\mathcal{P}(t)\mathcal{U}_{t_0}^t \mathcal{P}(t))^{(-1)}, \label{eq:Kdef}\\
				\mathcal{I}(t) &\equiv \dot{\mathcal{P}}(t) \mathcal{U}_{t_0}^t\mathcal{Q}(t)  + \mathcal{P}(t)\dot{\mathcal{U}}_{t_0}^t \mathcal{Q}(t) -  \mathcal{K}(t)  \mathcal{P}(t)\mathcal{U}_{t_0}^t \mathcal{Q}(t). \label{eq:Idef}
			\end{align}
		\end{theorem}

	\begin{proof}
		By direct calculation we obtain 
		\begin{align*}
			\frac{d}{dt}( \mathcal{P}(t)\mathcal{U}_{t_0}^t)  =&  \dot{\mathcal{P}}(t) \mathcal{U}_{t_0}^t + \mathcal{P}(t) \dot{\mathcal{U}}_{t_0}^t  \\
			=&  \dot{\mathcal{P}}(t) \mathcal{U}_{t_0}^t  \mathcal{P}(t)  +  \dot{\mathcal{P}}(t)  \mathcal{U}_{t_0}^t   \mathcal{Q}(t)+ \mathcal{P}(t) \dot{\mathcal{U}}_{t_0}^t  \mathcal{P}(t) +  \mathcal{P}(t) \dot{\mathcal{U}}_{t_0}^t  \mathcal{Q}(t) \\
			=&   \left(\dot{\mathcal{P}}(t) \mathcal{U}_{t_0}^t \mathcal{P}(t) + \mathcal{P}(t) \dot{\mathcal{U}}_{t_0}^t  \mathcal{P}(t)\right) (\mathcal{P}(t)\mathcal{U}_{t_0}^t \mathcal{P}(t))^{(-1)} \mathcal{P}(t)\mathcal{U}_{t_0}^t \mathcal{P}(t)   \\
			&  +  \dot{\mathcal{P}}(t)  \mathcal{U}_{t_0}^t  \mathcal{Q}(t) + \mathcal{P}(t)\dot{\mathcal{U}}_{t_0}^t \mathcal{Q}(t) \\
			=& \mathcal{K}(t) \mathcal{P}(t)\mathcal{U}_{t_0}^t \mathcal{P}(t)  +  \dot{\mathcal{P}}(t)  \mathcal{U}_{t_0}^t  \mathcal{Q}(t) + \mathcal{P}(t)\dot{\mathcal{U}}_{t_0}^t \mathcal{Q}(t)\\
			=& \mathcal{K}(t) \mathcal{P}(t)\mathcal{U}_{t_0}^t  +	\mathcal{I}(t) \mathcal{Q}(t) .
		\end{align*}
		Multiplying by $\rho(t_0)$, we have
		\begin{equation}\label{eq:masterEquationAlmost}
			\frac{d}{dt}( \mathcal{P}(t)\mathcal{U}_{t_0}^t \rho(t_0))   = \mathcal{K}(t) \mathcal{P}(t)\mathcal{U}_{t_0}^t\rho(t_0)  +	\mathcal{I}(t) \mathcal{Q}(t)\rho(t_0).
		\end{equation}
		The solution of equation~\eqref{eq:basicDiffEq} with a given initial condition $\rho(t_0)$ is unique \cite[Theorem 5.2]{Coddington1955} and can be represented in the form $\rho(t) = \mathcal{U}_{t_0}^t\rho(t_0)$, where $\mathcal{U}_{t_0}^t$ is defined by formula \eqref{eq:propEq}, so from \eqref{eq:masterEquationAlmost} we get \eqref{eq:masterEquation}.
	\end{proof}
	
	The integro-differential equation of the Nakajima--Zwanzig type in the case of a time-dependent projector of general form was considered in \cite{Fick1990} and therefore is sometimes called the Fick-Sauermann equation \cite{Kato2004}. This expression, is a time-local (time-convolutionless) analogue of the Fick-Sauermann equation. Moreover, like \cite[Section 9.2]{BreuerPetruccione2007} it could be derived by eliminating the time convolution from the Fick-Sauermann equation, but we follow the derivation from \cite{Karasev2023} which does not use the integro-differential equation as an intermediate step at all.
	
	\begin{corollary}\label{cor:pertExp}
		For  $\mathcal{K}(t) $ at  $\lambda \rightarrow 0 $ the following asymptotic expansion holds
		\begin{equation}\label{eq:expansionOfK}
			\mathcal{K}(t) = \sum_{n=1}^{\infty} \lambda^n \mathcal{K}_n(t),
		\end{equation}
		whose coefficients are determined by the formulas
		\begin{equation}\label{eq:coeffKn}
			\mathcal{K}_n(t) =  \sum_{q=0}^{n-1} (-1)^q \sum_{\sum_{j=0}^q k_j =n, k_j \geqslant 1}  \check{\mathcal{M}}_{k_0}(t) \mathcal{M}_{k_1}(t)  \ldots \mathcal{M}_{k_{q}}(t)
		\end{equation}
		(the sum is taken over all compositions of the number $n$), where
		\begin{align*}
			\mathcal{M}_{k}(t) \equiv& \int_{t_0}^t d t_1  \ldots  \int_{t_0}^{t_{k-1}} d t_k  \mathcal{P}(t)  \mathcal{L}(t_1) \ldots \mathcal{L}(t_k)\mathcal{P}(t),  \\
			\check{\mathcal{M}}_{k}(t) \equiv& \int_{t_0}^t d t_1  \ldots  \int_{t_0}^{t_{k-1}} d t_k  \dot{\mathcal{P}}(t)  \mathcal{L}(t_1) \ldots \mathcal{L}(t_k)\mathcal{P}(t) \\
			& + \int_{t_0}^t d t_1  \ldots  \int_{t_0}^{t_{k-2}} d t_{k-1}  \mathcal{P}(t) \mathcal{L}(t) \mathcal{L}(t_1) \ldots \mathcal{L}(t_{k-1})\mathcal{P}(t),
		\end{align*}
		and for $\mathcal{I}(t)$ at $\lambda \rightarrow 0 $ the asymptotic expansion holds
		\begin{equation}\label{eq:expansionOfI}
			\mathcal{I}(t)  = \sum_{n=1}^{\infty} \lambda^n \mathcal{I}_n(t),
		\end{equation}
		whose coefficients are determined by the formulas
		\begin{equation}\label{eq:coeffIn}
			\mathcal{I}_n(t) = \check{\tilde{\mathcal{M}}}_{n}(t) +  \sum_{q=1}^{n-1} (-1)^q \sum_{\sum_{j=0}^q k_j =n, k_j \geqslant 1}  \check{\mathcal{M}}_{k_0}(t) \mathcal{M}_{k_1}(t)  \ldots \mathcal{M}_{k_{q-1}}(t)\tilde{ \mathcal{M}}_{k_{q}}(t),
		\end{equation}
		where
		\begin{align*}
			\tilde{\mathcal{M}}_k(t) \equiv& \int_{t_0}^t d t_1  \ldots  \int_{t_0}^{t_{k-1}} d t_k  \mathcal{P}(t)  \mathcal{L}(t_1) \ldots \mathcal{L}(t_k)\mathcal{Q}(t),\\
			\check{\tilde{\mathcal{M}}}_{k}(t) \equiv& \int_{t_0}^t d t_1  \ldots  \int_{t_0}^{t_{k-1}} d t_k  \dot{\mathcal{P}}(t)  \mathcal{L}(t_1) \ldots \mathcal{L}(t_k)\mathcal{Q}(t) \\
			& + \int_{t_0}^t d t_1  \ldots  \int_{t_0}^{t_{k-2}} d t_{k-1}  \mathcal{P}(t) \mathcal{L}(t) \mathcal{L}(t_1) \ldots \mathcal{L}(t_{k-1})\mathcal{Q}(t).
		\end{align*}
	\end{corollary}
	
	\begin{proof}
		For the most part, the proof follows the proof of Theorem 1 from \cite{Karasev2023}, so we will cite only those places where there are differences. In the case of a constant projection $\mathcal{P}(t) = \mathcal{P}$, equation \eqref{eq:Kdef} reduces to equation (12) from \cite{Karasev2023} and can be obtained from it by the formal substitution $\mathcal{P}\dot{\mathcal{U}}_{t_0}^t \mathcal{P} \rightarrow \dot{\mathcal{P}}(t) \mathcal{U}_{t_0}^t\mathcal{P}(t) + \mathcal{P}(t)\dot{\mathcal{U}}_{t_0}^t \mathcal{P}(t) $. Taking into account the expansion of $\mathcal{U}_{t_0}^t$ into a Dyson series, which is also called the Peano-Baker series  in the case of finite-dimensional matrices  \cite{Baake2011},
		\begin{equation*}
			\mathcal{U}_{t_0}^t = \sum_{k=0}^{\infty} \lambda^k \int_{t_0}^t d t_1  \ldots  \int_{t_0}^{t_{k-1}} d t_k \mathcal{L}(t_1) \ldots \mathcal{L}(t_k)
		\end{equation*}
		the terms of the expansion $\mathcal{P}\dot{\mathcal{U}}_{t_0}^t \mathcal{P}$ in the Maclaurin series are replaced by the corresponding terms of the expansion $ \dot{\mathcal{P}}(t) \mathcal{U}_{t_0}^t\mathcal{P}(t) + \mathcal{P}(t)\dot{\mathcal{U}}_{t_0}^t \mathcal{P}(t )$: $ \dot{\mathcal{M}}_{k_0}(t)|_{\mathcal{P}(t) = \mathcal{P}} \rightarrow \check{\mathcal{M}}_ {k_0}(t) $. After such a replacement, formula (3) from \cite{Karasev2023} will take form \eqref{eq:coeffKn}. Expressions \eqref{eq:coeffIn} can be obtained similarly. 
	\end{proof}
	
	Dynamics with a time-dependent projector of the general form is quite rare in the literature. First of all, the case of the so-called Robertson dynamics \cite{Kato2000} is used. Namely, if $\rho(t)$ and $\mathcal{P}(t)$ at $t \in [t_0, + \infty)$ are such that
	\begin{equation}\label{eq:RobertsonType}
		\dot{\mathcal{P}}(t) \rho(t) = 0,
	\end{equation}
	then $\mathcal{P}(t) \rho(t)$ is said to evolve according to Robertson dynamics. In this case, $\dot{\mathcal{P}}(t) \rho(t) = \dot{\mathcal{P}}(t)\mathcal{U}_{t_0}^t \rho(t_0)$, so the terms including $\dot{\mathcal{P}}(t)$ in \eqref{eq:Kdef}--\eqref{eq:Idef} disappear. Similarly, in that case the terms $\check{\mathcal{M}}_{k}(t) $ and $\check{\tilde{\mathcal{M}}}_{k}(t) $ take on the form as if we had differentiated $\mathcal{M}_{k}(t)$ and $\tilde{\mathcal{M}}_{k}(t)$, respectively, but considered the projector $\mathcal{P}(t)$ independent of time $ t $. 
	
	Usually, for condition \eqref{eq:RobertsonType} to be satisfied, the dependence $\mathcal{P}(t)$ cannot be chosen ''globally'' such that it is satisfied for $\rho(t)$ whose dynamics is determined by different $\mathcal{L}(t) $. Instead, it is chosen as a composite function of $\rho(t)$ itself
	\begin{equation}\label{eq:timeDepAndNonLin}
		\mathcal{P}(t) = \mathcal{P}_{NL}(\rho(t)),
	\end{equation}
	where $\mathcal{P}_{NL}(\rho)$ is a function that is the nonlinear analog of the projector from the ''standard'' scheme of deriving master equations in open quantum systems theory \cite[Section 9]{BreuerPetruccione2007}. However, this is what time-dependent projection methods in nonequilibrium statistical physics are typically used for \cite[Section 17.3.2]{Fick1990}, \cite{Los2022}: They are interesting not by themselves, but as an intermediate ''trick'' for deriving nonlinear equations.
	
	In the next section we will also develop a similar approach by considering an important example where \eqref{eq:RobertsonType} is satisfied. However, we note that equations \eqref{eq:Kdef}, \eqref{eq:Idef} with the projector directly dependent on time, and without condition \eqref{eq:RobertsonType} may also be important in the theory of open quantum systems. In particular, in the case of the tensor product of two Hilbert spaces $\mathcal{H}_S \otimes \mathcal{H}_B$ it is quite natural to consider a time-dependent Argyres--Kelley projector of the form
	\begin{equation}\label{eq:AKproj}
		\mathcal{P}_{AK}(t) X = \Tr_B X \otimes \rho_B(t),
	\end{equation}
	where $\rho_B(t)$ is the density matrix in $\mathcal{H}_B$ in given time-dependent form to derive time-dependent master equations which are important in modern quantum theory problems of incoherent control \cite{Lokutsievskiy2021, Morzhin2021, Petruhanov2023}. In particular, in \cite{Accardi2006} the authors actually consider a reservoir that is in a time-dependent coherent state. However, such dynamics is assumed to be slow in the interaction representation, allowing $\dot{\mathcal{P}}_{AK}(t)$ to be neglected. However, there is also interest in the literature situations where such an approximation is not fulfilled \cite{Szczygielski2013, Mori2023}.

	\section{Generalized Kawasaki--Gunton projector}
	\label{sec:genKavGun}
	
	The definition of the Kawasaki--Gunton projector explicitly includes a family of Gibbs-type operators. We generalize the construction of the Kawasaki--Gunton projector based on the following definition.
	
	\begin{definition}
		Let a finite number of self-adjoint matrices $P_m \in \mathbb{C}^{d \times d}$, $m=1, \ldots, M$ linearly independent of each other and the identity matrix be given. These matrices will be called relevant observables. Let us also have a family $d \times d$ of matrices $\rho_{ans}(\vec{E})$, continuously differentiable dependent on the parameters $\vec{E}$ belonging to the domain in $\mathbb{R}^M$, and satisfying the conditions
		\begin{equation}\label{eq:consistCond}
			\Tr P_m \rho_{ans}(\vec{E}) = E_m,
		\end{equation}
		which we call the consistency conditions. This family $\rho_{ans}(\vec{E})$ will be called ansatz consistent with relevant observables $P_m$. 
	\end{definition}
	
	It will be convenient for us to introduce a vector $\vec{P}$ consisting of operators $P_m$. Then consistency conditions \eqref{eq:consistCond} will be written as
	\begin{equation}\label{eq:consistCondVec}
		\Tr \vec{P} \rho_{ans}(\vec{E}) = \vec{E},
	\end{equation}
	where multiplication of vector with matrix elements by matrix and trace taking are understood in an element-by-element sense.
	
	The standard definition of the Kawasaki--Gunton projector is based \cite{Kato2000} on a family of Gibbs distributions of the form
	\begin{equation}\label{eq:GibbsFamily}
		\rho_{Gibbs}(\vec{\beta}) = \frac{e^{- (\vec{\beta}, \vec{P})}}{Z(\vec{\beta})}, \qquad Z(\vec{\beta}) \equiv \Tr e^{- (\vec{\beta}, \vec{P})},
	\end{equation}
	where $ (\vec{\beta}, \vec{P}) \equiv \sum_{m} \beta_m P_m$. Then by solving the system 
	\begin{equation}\label{eq:consistCondGibbs}
		\Tr \vec{P} \rho_{Gibbs}(\vec{\beta}) = \vec{E}
	\end{equation}
	with respect to $\vec{\beta}$, where $\vec{E}$ acts as a parameter set, a function $\vec{\beta}(\vec{E})$ is constructed and the Gibbs family of the \eqref{eq:GibbsFamily}, reparameterized with parameters $\vec{E}$, acts as the ansatz consistent with the set $\vec{P}$, that is in this case
	\begin{equation*}
		\rho_{ans}(\vec{E}) = \rho_{Gibbs}(\vec{\beta}(\vec{E})).
	\end{equation*}
	The fact that $\vec{\beta}(\vec{E})$ is defined by system \eqref{eq:consistCondGibbs} automatically ensures that the consistency conditions of \eqref{eq:consistCond} are met. In addition, system \eqref{eq:consistCondGibbs} can be rewritten as
	\begin{equation*}
		\vec{E} = -\frac{\partial}{\partial \vec{\beta}} \ln Z(\vec{\beta}),
	\end{equation*}
	using only the partition function $ Z(\vec{\beta})$, which sometimes simplifies the calculation.
	
	Note that although in statistical physics the $P_m$ are often chosen as commuting observables, such as the Hamiltonian and the number of particles of a many-particle system, the situation where such relevant observables do not commute is also of sufficient interest. In particular, they arise on the Jaynes principle \cite{Jaynes1957} as states with maximum von Neumann entropy with fixed averages of the relevant observables and play an important role in modern quantum thermodynamics problems \cite{Mingo2018}. In \cite{Semin2020}  families of operators corresponding to the maximization of Renyi \cite{Baskirove2006, Bakiev2020} entropy under the same conditions were also considered as $\rho_{ans}(\vec{E}) $:
	\begin{equation}\label{eq:RenyiAns}
		\rho_{Renyi, q}(\vec{E}) = \frac{1}{Z_q(\vec{E})} \left(1 + \frac{q-1}{q} (\vec{\beta}(\vec{E}),\vec{E} -\vec{P} )\right)^{\frac{1}{q-1}}, 
	\end{equation}
	where $q$ is the Renyi entropy parameter, $Z_q(\vec{E})$ is defined by the normalization condition
	\begin{equation*}
		\Tr \rho_{Renyi, q}(\vec{E}) = 1,
	\end{equation*}
	and the functions $\vec{\beta}(\vec{E})$ are chosen such that consistency conditions \eqref{eq:consistCondVec} are met.

	Often the Kawasaki--Gunton projector is immediately defined as being explicitly time dependent. However, this dependence is essentially not explicit and depends on the dynamics of the density matrix. Therefore, let us define the Kawasaki--Gunton projector in 3 steps to make its definition, in our opinion, clearer. 
	
	\begin{definition}\label{def:KGprojector}

		\begin{enumerate}
			Let an ansatz $ \rho_{ans}(\vec{E})$ consistent with the relevant observables $\vec{P}$ be given. 
			\item Introduce a family of (linear) superoperators $\mathcal{P}_{KG,par}(\vec{E})$, also parameterized by vector~$\vec{E}$, which act on an arbitrary matrix $X$ by the formula
			\begin{equation}\label{eq:KGperDef}
				\mathcal{P}_{KG,par}(\vec{E}) X \equiv \rho_{ans}(\vec{E}) \Tr X +\left(\Tr (X \vec{P})- (\Tr X) \vec{E},\frac{\partial \rho_{ans} (\vec{E})}{\partial \vec{E}}\right).
			\end{equation}
			\item Introduce a super-operator-valued function $\mathcal{P}_{KG,NL}(\rho)$ of the density matrix $\rho$ by formula
			\begin{equation}\label{eq:KGNLDef}
				\mathcal{P}_{KG,NL}(\rho) = \mathcal{P}_{KG,par}(\vec{E}) |_{\vec{E} \equiv \Tr \vec{P} \rho}
			\end{equation}
			\item Let the density matrix $\rho(t)$ satisfy an equation of the form \eqref{eq:basicDiffEq}, then define a super-operator-valued time function $ t\in [t_0, +\infty)$
			\begin{equation}\label{eq:KGDef}
				\mathcal{P}_{KG}(t)  \equiv \mathcal{P}_{KG, NL}(\rho(t)).
			\end{equation}
			This function will be called the generalized Kawasaki--Gunton projector corresponding to the ansatz $ \rho_{ans}(\vec{E})$ consistent with the relevant observables $\vec{P}$. 
		\end{enumerate}
		
	\end{definition}

	Just in case, let us clarify that the scalar product in formula \eqref{eq:KGperDef} is simply the notation of the following expression
	\begin{equation*}
		\mathcal{P}_{KG,par}(\vec{E}) X \equiv \rho_{ans}(\vec{E}) \Tr X +\sum_{m}(\Tr X P_m - (\Tr X) E_m)\frac{\partial \rho_{ans} (\vec{E})}{\partial E_m}.
	\end{equation*}
	
	It is clear that in the last step this definition follows formula \eqref{eq:timeDepAndNonLin} adhering the general line of the approach to the derivation of nonlinear master equations discussed in the previous section.
	
	Let us check some properties of the generalized Kawasaki--Gunton projector, which generalize the properties of the traditional Kawasaki--Gunton projector operator.
	\begin{proposition}
		For arbitrary $\vec{E}$ and $\vec{E}'$
		\begin{equation}\label{eq:projLikeProp}
			\mathcal{P}_{KG,par}(\vec{E}) \mathcal{P}_{KG,par}(\vec{E}') = \mathcal{P}_{KG,par}(\vec{E}).
		\end{equation}
		
	\end{proposition}
	
	\begin{proof}
		Due to the fact that $\rho_{ans} (\vec{E}')$ are density matrices and in particular have a unit trace and conditions \eqref{eq:consistCondVec}, substituting in \eqref{eq:KGperDef}, we have
		\begin{align*}
			\mathcal{P}_{KG,par}(\vec{E}) \rho_{ans} (\vec{E}') &=  \rho_{ans}(\vec{E}) \Tr  \rho_{ans} (\vec{E}') +\left(\Tr ( \rho_{ans} (\vec{E}') \vec{P})- (\Tr  \rho_{ans} (\vec{E}')) \vec{E},\frac{\partial \rho_{ans} (\vec{E})}{\partial \vec{E}}\right)\\
			&=  \rho_{ans}(\vec{E}) + \left(\vec{E}' - \vec{E},  \frac{\partial \rho_{ans} (\vec{E})}{\partial \vec{E}}\right).
		\end{align*}
		In addition, note that
		\begin{align*}
			\mathcal{P}_{KG,par}(\vec{E})\frac{\partial \rho_{ans} (\vec{E}')}{\partial \vec{E}'} &=  \frac{\partial}{\partial \vec{E}'} (	\mathcal{P}_{KG,par}(\vec{E}) \rho_{ans} (\vec{E}'))\\
			&= \frac{\partial}{\partial \vec{E}'} \left(\rho_{ans}(\vec{E}) + \left(\vec{E}' - \vec{E},  \frac{\partial \rho_{ans} (\vec{E})}{\partial \vec{E}}\right)\right) =  \frac{\partial \rho_{ans} (\vec{E})}{\partial \vec{E}}.
		\end{align*}
		Then, according to the definition of \eqref{eq:KGperDef} taking into account linearity $\mathcal{P}_{KG,par}(\vec{E})$, we have
		\begin{align*}
			&\mathcal{P}_{KG,par}(\vec{E}) \mathcal{P}_{KG,par}(\vec{E}') X \\
			&= \mathcal{P}_{KG,par}(\vec{E})\left(\rho_{ans}(\vec{E}') \Tr X +\left(\Tr (X \vec{P})- (\Tr X) \vec{E}',\frac{\partial \rho_{ans} (\vec{E}')}{\partial \vec{E}'}\right)\right)\\
			&=\mathcal{P}_{KG,par}(\vec{E}) \rho_{ans}(\vec{E}') \Tr X +\left(\Tr (X \vec{P})- (\Tr X) \vec{E}', \mathcal{P}_{KG,par}(\vec{E})\frac{\partial \rho_{ans} (\vec{E}')}{\partial \vec{E}'}\right)\\
			&= \left(\rho_{ans}(\vec{E}) + \left(\vec{E}' - \vec{E},  \frac{\partial \rho_{ans} (\vec{E})}{\partial \vec{E}}\right)\right) \Tr X   +\left(\Tr (X \vec{P})- (\Tr X) \vec{E}', \frac{\partial \rho_{ans} (\vec{E})}{\partial \vec{E}}\right)\\
			&=  \rho_{ans}(\vec{E}) \Tr X +\left(\Tr (X \vec{P})- (\Tr X) \vec{E},\frac{\partial \rho_{ans} (\vec{E})}{\partial \vec{E}}\right) = \mathcal{P}_{KG,par}(\vec{E}) X.
		\end{align*}
		The result is \eqref{eq:projLikeProp}.
	\end{proof}

	In particular, if you put $ \vec{E} = \vec{E}' $ in \eqref{eq:projLikeProp}, then it will take the form
	\begin{equation*}
		(\mathcal{P}_{KG,par}(\vec{E}))^2 = \mathcal{P}_{KG,par}(\vec{E}),
	\end{equation*}
	i.e. $\mathcal{P}_{KG,par}(\vec{E})$ is idempotent at every fixed $\vec{E} $. Note that this, in principle, opens up the possibility of using it to derive master equations directly as a time-independent projector.
	
	Besides, from \eqref{eq:projLikeProp} directly follows the characteristic Kawasaki--Gunton operator \cite[Equation (2.3.29)]{Zubarev1997} property 
	\begin{equation*}
		\mathcal{P}_{KG}(t) \mathcal{P}_{KG}(t') = \mathcal{P}_{KG}(t)
	\end{equation*}
	for arbitrary $t, t'\in [t_0, +\infty)$. It is interesting to note that the time-dependent Argyres--Kelley projectors \eqref{eq:AKproj} also satisfy this property. 
	
	\begin{proposition}\label{prop:ans}
		If $\Tr \rho = 1$ (in particular, if $\rho$ is a density matrix)
		\begin{equation}\label{eq:projDenMat}
			\mathcal{P}_{KG,NL}(\rho) \rho = \rho_{ans}(\vec{E})_{\vec{E} = \Tr (\rho \vec{P}) } .
		\end{equation}
	\end{proposition}
	
	\begin{proof}
		Taking into account $\Tr \rho = 1$, the definition of \ref{def:KGprojector} takes the following form 
		\begin{align*}
			\mathcal{P}_{KG,par}(\vec{E}) |_{\vec{E} = \Tr (\rho \vec{P})} \rho = & \rho_{ans}(\vec{E})|_{\vec{E} = \Tr (\rho \vec{P})} \Tr \rho\\
			&
			+\left(\Tr (\rho \vec{P})- (\Tr \rho)(\Tr (\rho \vec{P})), \left. \frac{\partial \rho_{ans} (\vec{E})}{\partial \vec{E}}  \right|_{\vec{E} = \Tr (\rho \vec{P})}\right)
			= \rho_{ans}(\vec{E})  |_{\vec{E} = \Tr (\rho \vec{P})}
		\end{align*}
		As a result, we got \eqref{eq:projDenMat}.
	\end{proof}
	
	Thus, the function $\mathcal{F}(\rho) \equiv \mathcal{P}_{KG,NL}(\rho) \rho$ maps an arbitrary density matrix to a density matrix corresponding to the selected ansatz $ \rho_{ans}(\vec{E})$ such that the averages for all relevant observables of the original density matrix $\rho$ and the ansatz $ \rho_{ans}(\vec{E})$ coincide. 
	
	\begin{remark}\label{rem:zeroTrace}
		Note that in the case of traceless matrices $\Tr X = 0$ the expressions defining the Kawasaki--Gunton operator are simplified to
		\begin{equation*}
			\mathcal{P}_{KG,par}(\vec{E}) X = \left(\Tr \vec{P}  X  , \frac{\partial \rho_{ans}(\vec{E})}{\partial \vec{E}} \right).
		\end{equation*}
		In particular, this is fulfilled if $X = \mathcal{L}(t) Y$, where $\mathcal{L}(t)$ is the generator of the preserving trace dynamics, or if $X = \frac{d}{dt} \rho(t)$ at some time $t$, if $\Tr \rho(t) = 1$ at any time and $\rho(t)$ is a differentiable function of time $t$.
	\end{remark}
	
	\begin{proposition}
		Let $ \rho(t)$ be the density matrix at any $t \in [t_0, +\infty)$, then condition \eqref{eq:RobertsonType} is satisfied.
	\end{proposition}
	
	\begin{proof}
		With remark \ref{rem:zeroTrace} we have
		\begin{align*}
			\mathcal{P}_{KG}(t) \frac{d}{dt}\rho(t) &= \left(\Tr \vec{P}  \frac{d}{dt}\rho(t) , \frac{\partial \rho_{ans}(\vec{E})}{\partial \vec{E}} \right)_{\vec{E}  = \Tr (\rho(t) \vec{P})} 
			=\left( \frac{d}{dt} \vec{E},\frac{\partial \rho_{ans} (\vec{E})}{\partial \vec{E}}\right)_{\vec{E}  = \Tr (\rho(t) \vec{P})} 
			\\
			&= \frac{d}{dt} \left(\rho_{ans} (\vec{E})|_{\vec{E}  = \Tr (\rho(t) \vec{P})}\right)=  \frac{d}{dt}(\mathcal{P}_{KG}(t)\rho(t))
			=  \dot{\mathcal{P}}_{KG}(t) \rho(t)+ \mathcal{P}_{KG}(t)\frac{d}{dt}\rho(t)
		\end{align*}
		Reducing the terms $ \mathcal{P}_{KG}(t)\frac{d}{dt}\rho(t)$, we get \eqref{eq:RobertsonType}.
	\end{proof}
	
	Thus, we can use corollary \ref{cor:pertExp} in the case of Robertson dynamics, i.e. omitting terms with $\dot{\mathcal{P}}(t)$.
	
	\begin{theorem}\label{th:secOrderEq}
		Equation \eqref{eq:masterEquation}, omitting the terms above second order on $\lambda$ and assuming that the initial condition is consistent with the projector, in the case of the generalized Kawasaki--Gunton projector takes the form
		\begin{align}
			\frac{d}{dt}\left(\rho_{ans}(\vec{E})|_{\vec{E} = \Tr \rho(t) \vec{P}} \right)  =& \lambda\left(\Tr \left( \vec{P}\mathcal{L}(t)\rho_{ans} (\vec{E}) \right)   , \frac{\partial \rho_{ans}(\vec{E})}{\partial \vec{E}} \right)_{\vec{E} = \Tr \rho(t) \vec{P}} \nonumber	\\
			&+ \lambda^2 \left(\Tr \left( \vec{P} \mathcal{L}(t)  \int_{t_0}^t dt_1 \mathcal{L}(t_1)\rho_{ans} (\vec{E})   \right), \frac{\partial \rho_{ans}(\vec{E})}{\partial \vec{E}} \right)_{\vec{E} = \Tr \rho(t) \vec{P}} \nonumber \\
			& 	-  \lambda^2\left(\Tr \left(  \vec{P} \mathcal{L}(t)\rho_{ans} (\vec{E}) \right) , \frac{\partial \rho_{ans}(\vec{E})}{\partial \vec{E}} \right)_{\vec{E} = \Tr \rho(t) \vec{P}} \times \nonumber\\
			& \qquad \qquad\times
			\left(\Tr \left(  \vec{P}\int_{t_0}^t dt_1\mathcal{L}(t_1)\rho_{ans} (\vec{E}) \right)  , \frac{\partial \rho_{ans}(\vec{E})}{\partial \vec{E}} \right)_{\vec{E} = \Tr \rho(t) \vec{P}} . \label{eq:secOrderEq}
		\end{align}
	\end{theorem}

	\begin{proof}
		Taking into account remark \ref{rem:zeroTrace} and proposition \ref{prop:ans}, we have (in terms of corollary~\ref{cor:pertExp})
		\begin{align*}
			\check{\mathcal{M}}_{1}(t)\rho(t)  &= \mathcal{P}_{KG}(t) \mathcal{L}(t) \mathcal{P}_{KG}(t) \rho(t) = \mathcal{P}_{KG}(t) \mathcal{L}(t) \rho_{ans}(\vec{E})_{\vec{E} = \Tr \rho(t) \vec{P}} \\
			&= \left(\Tr (\vec{P} \mathcal{L}(t)\rho_{ans} (\vec{E}) )  , \frac{\partial \rho_{ans}(\vec{E})}{\partial \vec{E}} \right)_{\vec{E} = \Tr \rho(t) \vec{P}} ,	
			\\
			\mathcal{M}_{1}(t)\rho(t) &= \mathcal{P}_{KG}(t) \int_{t_0}^t dt_1 \mathcal{L}(t_1) \mathcal{P}_{KG}(t) \rho(t)\\
			&=  \left(\Tr \left(\vec{P}  \int_{t_0}^t dt_1 \mathcal{L}(t_1) \rho_{ans} (\vec{E}) \right)  , \frac{\partial \rho_{ans}(\vec{E})}{\partial \vec{E}} \right)_{\vec{E} = \Tr \rho(t) \vec{P}} ,
			\\
			\mathcal{M}_{2}(t)\rho(t) &= \mathcal{P}_{KG}(t) \mathcal{L}(t)  \int_{t_0}^t dt_1   \mathcal{L}(t_1) \mathcal{P}_{KG}(t) \rho(t) \\
			&=  \left(\Tr \left( \vec{P}  \mathcal{L}(t) \int_{t_0}^t dt_1 \mathcal{L}(t_1)\rho_{ans} (\vec{E})   \right), \frac{\partial \rho_{ans}(\vec{E})}{\partial \vec{E}} \right)_{\vec{E} = \Tr \rho(t) \vec{P}} .
		\end{align*}
		According to formula \eqref{eq:coeffKn} we have
		\begin{equation*}
			\mathcal{K}_{1}(t)= \check{\mathcal{M}}_{1}(t), \qquad
			\mathcal{K}_{2}(t)= \check{\mathcal{M}}_{2}(t) -\check{\mathcal{M}}_{1}(t) \mathcal{M}_{1}(t).
		\end{equation*}
		As a result, by theorem \ref{th:masterEquation} with corollary~\ref{cor:pertExp} we get \eqref{eq:secOrderEq}.
	\end{proof}
	
	If we confine ourselves to the first order of the expansion in $\lambda$, then, taking into account 
	\begin{equation*}
		\frac{d}{dt}  \rho_{ans}(\vec{E}(t)) = \left(\frac{d}{dt} \vec{E}(t), \frac{\partial \rho_{ans}(\vec{E})}{\partial \vec{E}}\right)_{\vec{E} = \Tr \rho(t) \vec{P}} 	,
	\end{equation*}
	equation \eqref{eq:secOrderEq} reduces to
	\begin{equation*}
		\frac{d}{dt} \vec{E}(t) = \lambda \Tr ( \vec{P} \mathcal{L}(t)\rho_{ans} (\vec{E}(t) ),
	\end{equation*}
	where $\vec{E}(t) = \Tr \rho(t) \vec{P}$.  This equation can be obtained by averaging \eqref{eq:basicDiffEq} over relevant observables, assuming that the exact density matrix $\rho(t)$ can be replaced by the ansatz $ \rho_{ans} (\vec{E}(t)) $. Thus, theorem \ref{th:secOrderEq} defines the following perturbative correction to such a naive approach. Similarly, using  theorem~\ref{th:masterEquation} and corollary~\ref{cor:pertExp}, one can obtain perturbative corrections of higher orders. 
	
	If we average both parts of equation \eqref{eq:secOrderEq} with $\vec{P}$, it can also be rewritten as an equivalent equation for averages from relevant observables $\vec{E}(t) = \Tr \rho(t) \vec{P}$:
	\begin{align}
		&\frac{d}{dt} \vec{E}(t) = \lambda   \Tr ( \vec{P}\mathcal{L}(t)\rho_{ans} (\vec{E}(t) ) + \lambda^2 \Tr \left( \vec{P} \mathcal{L}(t)  \int_{t_0}^t dt_1 \mathcal{L}(t_1)\rho_{ans} (\vec{E}(t))   \right) \nonumber \\
		& 	-  \lambda^2 \Tr\Biggl( \vec{P} \left(\Tr (  \vec{P} \mathcal{L}(t)\rho_{ans} (\vec{E}) ) , \frac{\partial \rho_{ans}(\vec{E})}{\partial \vec{E}} \right) \times \nonumber\\
		& \qquad \times \left(\Tr (  \vec{P}\int_{t_0}^t dt_1\mathcal{L}(t_1)\rho_{ans} (\vec{E}) )  , \frac{\partial \rho_{ans}(\vec{E})}{\partial \vec{E}} \right) \Biggr)_{\vec{E} =\vec{E}(t) }. \label{eq:secOrderEqE}
	\end{align}
	Both equation \eqref{eq:secOrderEq} and equation \eqref{eq:secOrderEqE} are non-linear in the general case.
	
	Note that theorem \ref{th:secOrderEq} also allows us to account for the case where the state is not initially consistent with the projector. This is important because some authors \cite[Section 2.4.2]{Zubarev1997} consider the assumption of an initial ansatz as a limitation of methods based on the Kawasaki--Gunton projector. However, theorem \ref{th:secOrderEq}, that we have obtained, allows us to get rid of this limitation.
	
	\section{Examples of projectors}
	\label{sec:examples}
	
	Firstly, we give an example which explains why, from our point of view, it is natural to parameterize the ansatz in the generalized Kawasaki--Gunton projector by averages of relevant observables. Consider the case of a 2-level system and assume that the family of distributions is given
	\begin{equation}\label{eq:twoLevelAns}
		\rho(\vec{\beta}) = \frac{1}{Z_f(\vec{\beta})} f((\vec{\beta}, \vec{\sigma})), \qquad Z_f(\vec{\beta}) \equiv \Tr f((\vec{\beta}, \vec{\sigma}))
	\end{equation}
	with parameters $\vec{\beta} \in \mathbb{R}^3$, where $f$ is an arbitrary function, $\vec{\sigma}$ is a vector composed of Pauli matrices \cite[page 20]{Holevo2018}. In particular, the cases of Gibbs form \eqref{eq:GibbsFamily} and Renyi form \eqref{eq:RenyiAns} fall here if the Pauli matrices $\vec{\sigma}$ or their subset are chosen as the relevant observables (here it will be convenient for us to still consider three-dimensional $\vec{\beta}$, but to assume that the corresponding components this vector are zero). 
	
	The function of the matrix $2 \times 2$ in terms of Pauli matrices can be calculated explicitly
	\begin{equation*}
		\rho(\vec{\beta}) =  \frac{1}{Z_f(\vec{\beta})} \left( I \frac{f(|\vec{\beta}|) + f(-|\vec{\beta}|)  }{2} + (\vec{\beta}, \vec{\sigma})   \frac{f(|\vec{\beta}|) - f(-|\vec{\beta}|)  }{2 |\vec{\beta}|}\right).
	\end{equation*}
	Then from the normalisation condition we obtain $Z_f(\vec{\beta}) = f(|\vec{\beta}|) + f(-||\vec{\beta}|)$
	and from the consistency conditions we get
	\begin{equation*}
		\vec{E} =\Tr \vec{\sigma}\rho(\vec{\beta}) = \vec{\beta}  \frac{f(|\vec{\beta}|) - f(-|\vec{\beta}|)  }{|\vec{\beta}| Z_f(\vec{\beta})}.
	\end{equation*}
	As a result, after reparametrization in terms of $\vec{E} $ we get
	\begin{equation*}
		\rho_{ans}(\vec{E}) = \rho(\vec{\beta}(\vec{E} ))= \frac12 \left(I + (\vec{E}, \vec{\sigma})\right)
	\end{equation*}
	(if in the original parametrization the component $\vec{\beta}$ is zeroed due to the absence of the corresponding Pauli matrix among the set of relevant observables, then the corresponding component $\vec{E}$ is also zeroed). Thus, the ansatz $\rho_{ans}(\vec{E})$ itself and its corresponding Kawasaki--Gunton projector will not depend on the function $f$ at all. However, the standard construction will set it in very different parameterizations \eqref{eq:GibbsFamily} and \eqref{eq:RenyiAns}.
	
	Moreover, if one chooses all three Pauli matrices as relevant observables, then the Kawasaki--Gunton projector reduces to the identical one. And if only some of the Pauli matrices are relevant, then the projector in terms of the Bloch vector leaves only those of its components that correspond to the relevant Pauli matrices. For example, if we choose $\sigma_x$ and $\sigma_z$ as relevant observables, then by a direct calculation we obtain
	\begin{equation}\label{eq:linExamp}
		\mathcal{P}_{KG,par}(\vec{E} ) X = \frac12 (\Tr X + \sigma_x \Tr( \sigma_x X) + \sigma_z \Tr( \sigma_z X)).
	\end{equation}
	Thus, the projector $\mathcal{P}_{KG,par}(\vec{E} ) $ will not depend on $\vec{E} $ and, consequently, the generalized Kawasaki--Gunton projector $\mathcal{P}_{KG}(t )$ will not depend on time.

	In general, if the dependence $\rho_{ans}(\vec{E})$ on $\vec{E}$ is linear, then $\mathcal{P}_{KG}(t ) = \operatorname{const} $. 
	\begin{proposition}
		Let the ansatz $\rho_{ans}(\vec{E}) $ be $ \rho_{ans}(\vec{E}) = B_0 + (\vec{E}, \vec{B})$, where $ B_0 $ and $\vec{B}$ constitute a fixed (independent of $\vec{E} $) set of matrices, then $\mathcal{P}_{KG}(t ) = \operatorname{const} $. 
	\end{proposition}
	
	\begin{proof}
		Taking into account
		\begin{equation*}
			\frac{\partial \rho_{ans} (\vec{E})}{\partial \vec{E}} = \vec{B},
		\end{equation*}
		and formula \eqref{eq:KGperDef} we have
		\begin{align*}
			\mathcal{P}_{KG,par}(\vec{E}) X &= ( B_0 + (\vec{E}, \vec{B}))\Tr X +\left(\Tr (X \vec{P})- (\Tr X) \vec{E},  \vec{B}\right) \\
			&= B_0  \Tr X +\left(\Tr (X \vec{P}),  \vec{B}\right) = \mathcal{P}_{KG,par}(\vec{0}) X.
		\end{align*}
		Thus, according to formulas \eqref{eq:KGNLDef} and \eqref{eq:KGDef} we get $\mathcal{P}_{KG}(t ) = \operatorname{const} $. 
	\end{proof}
	
	\begin{remark}
		In the ansatz $ \rho_{ans}(\vec{E}) = B_0 + (\vec{E}, \vec{B})$, the choice of matrices $B_0$ and $ \vec{B}$ is not arbitrary. Due to the self-conjugacy of $\rho_{ans}(\vec{E})$ and the realness of $\vec{E}$, $B_m$ must also be self-conjugate $ B_m = B_m^{\dagger}$. The normalization condition $ \Tr \rho_{ans}(\vec{E}) = 1$ and the consistency condition $ \Tr (\vec{P} \rho_{ans}(\vec{E})) = \vec{E}$ lead to the following conditions on these matrices
		\begin{equation*}
			\Tr B_{k} P_{m} = \delta_{km}
		\end{equation*}
		where $k, m =0, \ldots, M$ and the notation $P_0 = I$ is introduced, i.e. $B_k$ and $P_m$ are biorthogonal with respect to the Hilbert-Schmidt scalar product.
	\end{remark}

	In particular, one can parameterize the density matrices by a generalized Bloch vector (of dimension $d^2-1$ if the dimension of matrices is $d \times d$) and choose one or more generalized Gell--Mann matrices \cite[Section 2.4]{Alicki2007} as the relevant observables. The Argyres--Kelley projector can in fact also be seen as a special case of the constant-time generalized Kawasaki--Gunton projector if one chooses $\rho_S(\vec{E}) \otimes \rho_B$ as the ansatz, where $ \rho_S(\vec{E})$ is the system density matrix parameterized by its generalized Bloch vector, and select all generalized Gell-Mann matrices in the system tensor space multiplied by the identity matrix in the reservoir space as relevant observables. Thus, the generalized Kawasaki--Gunton projector introduced by us allows us to combine both ''dynamical'' and ''thermodynamical'' approaches to open quantum systems in terms of work~\cite{Semin2020}.

	In the case of a two-level system and one relevant observable $\sigma_z$ the most general, ansatz consistent with this observable, takes the form
	\begin{equation}\label{eq:nonLinAns}
		\rho_{ans}(E) = \frac12 \left(I + E \sigma_z + f( E) \sigma_x + g( E) \sigma_y \right),
	\end{equation}
	where $f$ and $g$ are arbitrary continuously differentiable functions of $E$. Indeed, one component of the Bloch vector is uniquely fixed by the consistency condition $ \Tr (\sigma_z \rho_{ans}(E)) = E $, and the remaining components can be chosen arbitrarily, but since we assume that the ansatz must be fully parameterized by the averages of the relevant observable $ \sigma_z $, they must be functions of $E$.
	
	In that case
	\begin{equation}\label{eq:nonLinAnsDiff}
		\frac{\partial \rho_{ans}(E)}{\partial E} = \frac12 \left( \sigma_z + f'(  E)  \sigma_x  + g'(  E)  \sigma_y \right)
	\end{equation}
	and
	\begin{align*}
		\mathcal{P}_{KG,par}(E)
		= &  \frac12 \biggl(I +  \Tr (X \sigma_z) \sigma_z \\
		&+ (f(  E) \Tr X +  (\Tr (X \sigma_z) - (\Tr X) E)  f'(  E)) \sigma_x \\
		& +  (g(  E) \Tr X +  (\Tr (X \sigma_z) - (\Tr X) E)  g'(  E)) \sigma_y \biggr).
	\end{align*}
	The projector $\mathcal{P}_{KG}(t)$ in this case is no longer reducible to a constant if the functions $f(E)$ and $g(E)$ are not linear.
	
	\section{Examples of equations and dissipative Wick rotation}
	\label{sec:examplesOfEq}
	
	As an example, consider the Gorini-Kossakowski-Sudarshan-Lindblad equation used to describe resonant fluorescence in a rotating frame \cite[Section 3.4.5]{BreuerPetruccione2007} 
	\begin{align} 
		\frac{d}{d t}\rho_{rot}(t) = \lambda \frac{i\Omega}{2}[\sigma_+ + \sigma_-, \rho_{rot}(t)] 
		&+ \gamma_0(N + 1) \left(\sigma_-\rho_{rot}(t)\sigma_+ - \frac{1}{2}\sigma_+\sigma_-\rho_{rot}(t)-\frac{1}{2}\rho_{rot} (t)\sigma_+\sigma_-\right) \nonumber \\
		&+ \gamma_0 N\left(\sigma_+\rho_{rot}(t)\sigma_- - \frac{1}{2}\sigma_-\sigma_+\rho_{I}(t) - \frac{1}{2}\rho_{rot}(t)\sigma_-\sigma_+\right) \label{eq:exampGKSL}
	\end{align}
	with projector \eqref{eq:linExamp}. To reduce this equation to form \eqref{eq:basicDiffEq}, we have to go to the ''interaction representation'' by making the substitution
	\begin{equation}\label{eq:intPicTransform}
		\rho(t) \equiv e^{ -\mathcal{L}_0 t } \rho_{rot}(t),
	\end{equation}
	where
	\begin{equation}\label{eq:freeGen}
		\mathcal{L}_0 \equiv \gamma_0(N + 1) \left(\sigma_- \; \cdot \; \sigma_+ - \frac{1}{2}\{\sigma_+\sigma_-, \; \cdot \; \}\right) + \gamma_0 N \left(\sigma_+ \; \cdot \; \sigma_- - \frac{1}{2}\{\sigma_-\sigma_+, \; \cdot \; \}\right).
	\end{equation}
	As a result, $ \rho(t)$ will satisfy equation \eqref{eq:basicDiffEq}, where
	\begin{equation}\label{eq:intGen}
		\mathcal{L}(t) = e^{-\mathcal{L}_0 t} \mathcal{L}_1 e^{\mathcal{L}_0 t}, \qquad \mathcal{L}_1 \equiv \frac{i\Omega}{2}[\sigma_+ + \sigma_-, \; \cdot \;] .
	\end{equation}
	Then, given the notation $\gamma \equiv \gamma_0 (2 N+1)$, if $t_0 = 0$, equations \eqref{eq:secOrderEqE} take the form
	\begin{equation}\label{eq:exampLinEqBeforSc}
		\frac{d}{d t}E_x(t) = 0, \qquad \frac{d}{dt}E_z(t) = 2 \lambda^2 \frac{\Omega^2}{\gamma} (1 - e^{  \frac{\gamma}{2} t}) E_z(t) + 2 \lambda^2  \frac{\gamma_0 \Omega^2}{\gamma^2} (1 - e^{  \frac{\gamma}{2} t})^2,
	\end{equation}
	where $ E_x(t) = \operatorname{Tr}\sigma_x \rho(t) $, $ E_z(t) = \operatorname{Tr}\sigma_z \rho(t) $. However, the solution of equation \eqref{eq:exampGKSL} can be found exactly \cite[Section 3.4.5]{BreuerPetruccione2007}, although it is quite cumbersome. As a result, using \eqref{eq:intPicTransform}, $E_x(t)$ and $E_z(t)$ can also be computed exactly.  We will denote these exact solutions by $E_x(t)^{\rm exact} $ and $E_z^{\rm exact} (t) $, keeping the notations $E_x(t)$ and $ E_z(t) $ for the solutions of the approximate equation \eqref{eq:exampLinEqBeforSc} with the same initial conditions. Then the direct computation gives:
		\begin{align}
			E_x(t) - E_x^{\rm exact}(t) =&0, \nonumber\\
			E_z(t) - E_z^{\rm exact}(t) =& - \lambda^4\frac{8}{3 \gamma^5} \biggl( 3 \gamma \left(5 + \gamma t - 2 e^{\frac{\gamma}{2}t}(2-\gamma t) -e^{\gamma t}\right) E_z(0) \nonumber\\
			&+ \gamma_0 \left(17 + 3 \gamma t- 9 e^{\frac{\gamma}{2}t}(1- \gamma t) -9 e^{\gamma t}  + e^{\frac{3\gamma}{2}t}\right) \biggr) + o(\lambda^4).
			\label{eq:errorForLinSecOrder}
		\end{align}
		Thus, the error of solutions of the equations \eqref{eq:exampLinEqBeforSc} is of order $O(\lambda^4)$. This often occurs in many common models in open systems theory, when the asymptotic expansion includes only even degrees of $\lambda$ \cite[Section 9.1.1]{BreuerPetruccione2007}.
	
	To distinguish the timescale at which the dynamics becomes consistent with the projector in the theory of open quantum systems one often uses the \cite{Accardi2002, Davies1974} Bogolubov-van Hove scaling $t \rightarrow \lambda^{-2} t$.  Let's introduce $\tilde{ E}_j(t) = E_j(\lambda^{-2} t) $, $j = x, z$, then
	\begin{equation}\label{eq:exampLinEq}
		\frac{d}{d t}\tilde{E}_x(t) = 0, \qquad \frac{d}{dt}\tilde{E}_z(t) = 2  \frac{\Omega^2}{\gamma} (1 - e^{  \frac{\gamma}{2} \frac{t}{\lambda^2}}) \tilde{ E}_z(t) + 2  \frac{\gamma_0 \Omega^2}{\gamma^2} (1 - e^{  \frac{\gamma}{2} \frac{t}{\lambda^2}})^2.
	\end{equation}
	Generally, in the theory of open quantum systems after the Bogolubov-van Hove scaling there are rapidly oscillating or decaying at $\lambda \rightarrow 0$ terms  in the coefficients of the second order equation \cite{Trushechkin2021, Teretenkov2021}. Here, on the contrary, the terms $e^{ \frac{\gamma}{2} \frac{t}{\lambda^2}}$ are exponentially increasing. However, formally assume that $\gamma <0$. Since $\gamma$ is proportional to the coupling constant of the open system to the environment, and becomes negative when this coupling constant is replaced by an imaginary one, we will call such a procedure a dissipative Wick rotation. Then at $\lambda \rightarrow 0$ the terms $e^{ \frac{\gamma}{2} \frac{t}{\lambda^2}}$ can be omitted and equations \eqref{eq:exampLinEq} take the form: 
	\begin{equation}\label{eq:exampLinEqLimit}
		\frac{d}{d t}\varepsilon_x(t) = 0, \qquad \frac{d}{dt}\varepsilon_z(t) = 2  \frac{\Omega^2}{\gamma} \varepsilon_z(t) + 2  \frac{\gamma_0 \Omega^2}{\gamma^2}.
	\end{equation}
	The possibility of such a transition in the case of growing exponents does not seem plausible. As mentioned above, the solution of equation \eqref{eq:exampGKSL} can be found exactly. Equations \eqref{eq:exampLinEqLimit} are also can be solved explicitly. A direct calculation of the limit in the exact solution of the equation \eqref{eq:exampGKSL} taking into account \eqref{eq:intPicTransform} gives
	\begin{equation*}
		\lim\limits_{\lambda \rightarrow + 0}\Tr \vec{P}\rho(\lambda^{-2}t) = \vec{\varepsilon}(t) .
	\end{equation*}
	We emphasise that we do not currently know how to justify discarding exponentially growing terms in the coefficients of a second order equation. However, a direct test in this particular example shows that such a discarding can be done in this case. This observation may be important for the theory of open quantum systems in general since in the standard formulation the existence of the Bogolubov-van Hove limit and its corrections are based \cite{Teretenkov2021, Trushechkin2021Long} on the attenuation of correlation functions of the reservoir. And this example gives some hope that the dissipative Wick rotation will allow one to compute a result that can be rigorously justified. Note also in the case of second order integrodifferential master equations such a dissipative Wick rotation allows one to ''collapse'' the kernel of the integral part of the equation into a delta function, in the situation when this kernel is not damped at large times.

	Now consider the high-temperature equivalent of generator \eqref{eq:freeGen}
	\begin{equation*}
		\mathcal{L}_0 \equiv \frac{\gamma}{2} \left(\sigma_- \; \cdot \; \sigma_+ - \frac{1}{2}\{\sigma_+\sigma_-, \; \cdot \; \}\right) + \frac{\gamma}{2}\left(\sigma_+ \; \cdot \; \sigma_- - \frac{1}{2}\{\sigma_-\sigma_+, \; \cdot \; \}\right).
	\end{equation*}
	with the same $\mathcal{L}_1 $ in  formula \eqref{eq:intGen} and the projector corresponding to  \eqref{eq:nonLinAns} in which we choose an arbitrary $ f(E)$ --- we will see that it will still not contribute to the second order equation, and $g(E) = \alpha \sqrt{E}$, where $\alpha$ is a real constant. This choice for the function $g(E)$ is natural since this relation between the components of the Bloch vector is fulfilled in the case of free dynamics with generator $\mathcal{L}_0$. Then, taking into account \eqref{eq:nonLinAnsDiff}, we have
	\begin{equation*}
		\Tr \sigma_z \frac{\partial \rho_{ans}(E)}{\partial E} \frac{\partial \rho_{ans}(E)}{\partial E} = 0,
	\end{equation*}
	so the last term in the right-hand side of equation \eqref{eq:secOrderEqE} will zero out. Furthermore, substituting~\eqref{eq:nonLinAns} into the other terms of equation \eqref{eq:secOrderEqE} given $g(E) = \alpha \sqrt{E}$, we get 

		\begin{align*}
			\Tr ( \sigma_z \mathcal{L}(t)\rho_{ans} (\vec{E}) ) &= - \alpha\sqrt{ E} \Omega e^{\frac{\gamma}{2} t},\\
			\Tr \biggl(  \sigma_z  \mathcal{L}(t)  \int_{t_0}^t dt_1 \mathcal{L}(t_1)\rho_{ans} (E)   \biggr) &= - 2 \gamma \frac{\Omega^2}{\gamma^2} \left(e^{\frac{\gamma}{2}t} - 1\right) E.
	\end{align*}
	Thus, equation \eqref{eq:secOrderEqE} will take the form
	\begin{equation}\label{eq:nonLinEqExam}
		\frac{d}{dt}E(t) =  - \lambda \alpha \sqrt{E(t)} \Omega e^{\frac{\gamma}{2} t} - 2 \lambda^2 \gamma \frac{\Omega^2}{\gamma^2} \left(e^{\frac{\gamma}{2}t} - 1\right) E(t).
	\end{equation}
	Given a fixed initial condition $E(0)$, it has two solutions
	\begin{align}
		E(t) =&  \exp \left( \frac{4 \lambda^2 \Omega^2}{\gamma^2} \left(e^{\frac{\gamma}{2}t} - 1- \frac{\gamma}{2} t\right)\right) \times \nonumber\\
		&\times \left(\sqrt{E(0)} \pm  \lambda\frac{\alpha \Omega }{2} \int_0^t d\tau \exp \left( \frac{\gamma}{2} \tau- \frac{2 \lambda^2 \Omega^2}{\gamma^2}\left(e^{ \frac{\gamma}{2} \tau} -1 -\frac{\gamma}{2} \tau\right)\right) d\tau\right)^2. \label{eq:nonLinSol}
	\end{align}
	This is the simplest example demonstrating the mechanism of nonlinear dynamics of averages for relevant observables in the theory of open quantum systems. However, we think it can be useful as a test problem for the development of approaches to more complex situations.
	
	Separately, we note that the fact that the function $f$ is not included in equation \eqref{eq:nonLinEqExam} de\-mon\-strates that the structure $\mathcal{L}(t)$ can exclude some ansatz parameters. And, as a result, different ansatz can lead to the same equations for the averages from relevant observables.
	
		Similar to what we did for equation \eqref{eq:exampLinEqBeforSc}, let us analyze the main order accuracy of $\lambda$ in equation \eqref{eq:nonLinEqExam}. At $\alpha \neq 0$ the leading order is the first order and \eqref{eq:nonLinEqExam} takes the form 
		\begin{equation}\label{eq:nonLinEqExamFirstOrder}
			\frac{d}{dt}E(t) =  - \lambda \alpha \sqrt{E(t)} \Omega e^{\frac{\gamma}{2} t}.
		\end{equation}
		Then solution \eqref{eq:nonLinSol} is simplified and takes the explicit form
		\begin{equation*}
			E(t) = E_{\pm}(t)   = \left(\sqrt{E(0)} \pm \lambda \frac{\alpha \Omega}{\gamma} (e^{\frac{\gamma}{2} t} -1)\right)^2.
		\end{equation*} 
		Using again the exact solution of \eqref{eq:exampGKSL} at $\gamma_0 = 0$, we have
		\begin{equation*}
			E_-(t) - E^{\rm exact}(t) = \lambda^2 \frac{\Omega^2}{\gamma^2} \left(\alpha^2 (e^{\frac{\gamma}{2} t} -1)^2-2(2 + \gamma t - 2e^{\frac{\gamma}{2} t})E(0)\right) + o(\lambda^2).
		\end{equation*}
		Thus, the error of the solution $E_-(t)$ of equation \eqref{eq:nonLinEqExamFirstOrder} is of order $O(\lambda^2)$, which is natural to expect if we wrote equation \eqref{eq:nonLinEqExam} neglecting the $O(\lambda^2)$ terms. However, we note that for $E_+(t)$ this is not satisfied
		\begin{equation*}
			E_+(t) - E^{\rm exact}(t) = E_+(t) - E_-(t) + o(\lambda^2)= O(\lambda).
		\end{equation*}
		In fact, $E_+(t)$ is a side solution of equation \eqref{eq:nonLinEqExamFirstOrder}. In general, since initial equations \eqref{eq:basicDiffEq} are linear, the dynamics of the mean values of the relevant observables $\vec{E}(t)$ is uniquely determined by the initial conditions. Therefore, if approximate nonlinear equations \eqref{eq:secOrderEq} have several solutions different within a given asymptotic accuracy, only one of them can be a correct asymptotic approximation of the exact solution.
		
		At $\alpha = 0$ the main order is the second order and equation \eqref{eq:nonLinEqExam} in the main order takes the following form
		\begin{equation*}
			\frac{d}{dt}E(t) =- 2 \lambda^2 \gamma \frac{\Omega^2}{\gamma^2} \left(e^{\frac{\gamma}{2}t} - 1\right) E(t).
		\end{equation*}
		This equation is also solved explicitly. As a result, we have
		\begin{equation*}
			E(t) - E^{\rm exact}(t) = - \lambda^4\frac{8}{\gamma^4} \left(5 + \gamma t - 2 e^{\frac{\gamma}{2}t}(2-\gamma t) -e^{\gamma t}\right) E(0),
		\end{equation*}
		which is the same as \eqref{eq:errorForLinSecOrder} if we put $\gamma_0 = 0$.
		
		Thus, at least in the examples considered, approximate equations have solutions whose error is limited by the terms we neglect in the equation. But in the case of nonlinear equations, only one of the solutions may have this property. However, obtaining any general estimates of accuracy is beyond the scope of this study.
	
	\section{Conclusion}
	
	In this paper, we have obtained perturbative expansions for the coefficients of linear kinetic equations local in time. We introduced a generalized Kawasaki-Ganton projector for an arbitrary ansatz parameterized by the averages of relevant observables and obtained a second-order equation for such a projector. We have considered and discussed a few examples of the introduced projectors and the derived equations. Note that in order to obtain higher order equations, the methods developed in this paper can be combined with other approaches known in open systems theory. In particular, in addition to the obtained form for expansion coefficients \eqref{eq:expansionOfK}, \eqref{eq:expansionOfI}, their expressions in terms of cumulants analogous to the Kubo-van Kampen cumulants \cite{Kubo1963, VanKampen1974I, VanKampen1974II}, their representation in a form analogous to \cite[Theorem 2]{Karasev2023}, as well as recurrence expressions for them analogous to those obtained in \cite{Gasbarri2018} may be of interest.
	
	Among the directions of further development, one of the most natural is the consideration of Gaussian ansatzes parameterized by first and second moments. In the case of a finite number of fermionic modes our results are directly applicable, in the bosonic case a generalization of these results to infinite dimensional operators is required, but at the physical level of rigor our results are also applicable. Irreversible quantum dynamics preserving Gaussian states has been well studied by \cite{Dodonov1986, Heinosaari2010, Teretenkov2019}, so it is natural enough in the case of perturbation of such dynamics to seek corrections in the assumption of Gaussian ansatz. In this case, equations \eqref{eq:secOrderEq} will map Gaussian states into Gaussian states, but will be nonlinear, so the corresponding evolutionary maps will be nonlinear generalizations of linear maps \cite{Palma2015} known in the literature, in particular, Gaussian channels \cite[Section 12.4]{Holevo2010}.
	
	By observables we meant everywhere self-adjoint matrices, but, probably, it is natural to abandon this, parameterizing the ansatzes by averages of generalized observables \cite[Section 2.1]{Holevo2003}, which will lead to some generally nonlinear analogue of quasi-probabilistic re\-pre\-sen\-tations \cite{Yashin2020}.
	
	It is also interesting to consider situations when the ansatz uses density matrices of smaller dimensionality as a parameter. In the case of nonlinear density matrix dependence an important issue is the comparison of equations of form \eqref{eq:secOrderEq} with nonlinear quantum equations, such as quantum nonlinear Boltzmann equations \cite[Section 3.7.1]{BreuerPetruccione2007}. Linear ansatzes of this kind have been studied in the literature; moreover, cases where the ansatz itself depends on a small parameter \cite{Trushechkin2021Der} have been considered. The development of our approach in this direction looks promising. This development, in particular, may be important for the derivation of effective generators \cite{Trubilko2019, Trubilko2020, Basharov2021, Teretenkov2022Eff} by projection methods and for obtaining their nonlinear analogue.

	The authors are grateful to G. Gasbarri, A.Yu. Karasev, E.O. Kiktenko, A.M. Savchenko, R.~Singh, and A.S. Trushechkin for discussion of the problems discussed in the paper. The authors are grateful to the reviewer for their valuable comments, which allowed us to significantly improve the text of the paper.

\end{document}